\documentclass[12pt]{article}

\usepackage{subfigure}

\usepackage{graphicx}

\usepackage{amsmath} 
\usepackage{amssymb} 
\usepackage{comment}

\newcommand{\bea}{\begin{eqnarray}}
\newcommand{\eea}{\end{eqnarray}}
\newcommand{\beq}{\begin{equation}}
\newcommand{\eeq}{\end{equation}}
\newcommand{\nn}{\nonumber}

\newcommand{\avg}[1]{\langle #1 \rangle_{\nu}}

\newcommand{\leaveout}[1]{}

\newcommand{\half}{{\frac{1}{2}}}

\newcommand{\dH}{{d_H}}

\newcommand{\dS}{{d_S}}
\newcommand{\dhh}{{d_h}}
\newcommand{\dhhb}{{\bar{d}_h}}
\newcommand{\ds}{{d_s}}

\newcommand{\C}[1]{{\mathcal{#1}}}

\newcommand{\Prob}{{\mathbb{P}}}

\newcommand{\G}[1]{\Gamma\left(#1\right)}

\newcommand{\measet}[1]{\nu( \{ #1\})}

\newtheorem{theorem}{Theorem}
\newtheorem{lemma}[theorem]{Lemma}
\newtheorem{proposition}[theorem]{Proposition}

\newenvironment{proof}[1][Proof]{\begin{trivlist}
\item[\hskip \labelsep {\bfseries #1}]}{\end{trivlist}}

\newenvironment{example}[1][Example]{\begin{trivlist}
\item[\hskip \labelsep {\bfseries #1}]}{\end{trivlist}}

\newenvironment{conjecture}[1][Conjecture]{\begin{trivlist}
\item[\hskip \labelsep {\bfseries #1}]}{\end{trivlist}}

\begin{document}

\thispagestyle{empty}
\begin{center}
\vspace{48pt}
{ \Large \bf Spectral dimension of trees \\ with a unique infinite spine}

\vspace{40pt}
{\sl Sigurdur \"O. Stef\'ansson}$^{1}$
and {\sl Stefan Zohren}$^{2}$  
\vspace{24pt}

{\small

$^1$~NORDITA, Roslagstullbacken 23,\\
ES-106 91 Stockholm, Sweden

\vspace{10pt}

$^2$~Rudolf Peierls Centre for Theoretical Physics, University of Oxford,\\
1 Keble Road, Oxford OX1 3NP, UK.

\vspace{10pt}

Email: {\tt sigste@nordita.org} and {\tt zohren@physics.ox.ac.uk}

\vspace{30pt}
}

\end{center}

\begin{abstract} 
Using generating functions techniques we develop a relation between the Hausdorff and spectral dimension of trees with a unique infinite spine. Furthermore, it is shown that if the outgrowths along the spine are independent and identically distributed, then both the Hausdorff and spectral dimension can easily be determined from the probability generating function of the random variable describing the size of the outgrowths at a given vertex, provided that the probability of the height of the outgrowths exceeding $n$ falls off as the inverse of $n$. We apply this new method to both critical non-generic trees and the attachment and grafting model, which is a special case of the vertex splitting model, resulting in a simplified proof for the values of the Hausdorff and spectral dimension for the former and novel results for the latter. 
\end{abstract}


\vspace{1cm}
Submitted to: {\em J.\ Stat.\ Phys.}


\newpage
\tableofcontents
\newpage

%
%

\section{Introduction}

Random trees and their properties have been studied in several branches of probability theory, mathematical physics and science in general. Their applications span from phylogenetic trees \cite{ald:1996,ford:2005} and random folding of RNA molecules \cite{david:2008} to models of quantum gravity \cite{Ambjorn:1997di,Sch:1998,leGall2005,DiFrancesco:2004qj}, amongst others. 

In general terms a random tree ensemble is a set of trees together with a probability measure associated to it. We  focus on a class of infinite, rooted trees which have the property that there is a unique, non-backtracking path from the root to infinity. This path will be referred to as {\it the spine}. In this case the ensemble is characterized by the distribution of the finite outgrowths from the spine. Such ensembles can arise in very different contexts: One example considered in this article is the equilibrium statistical mechanics model of \emph{simply generated trees}, first introduced in \cite{meir:1978}, which describes trees with a local action. These random tree ensembles are related to critical and sub-critial \emph{Galton-Watson processes} \cite{Athreya1972}. The second example considered is the \emph{attachment and grafting model} \cite{stefansson:2011a} which is a specific case of the \emph{vertex splitting model} \cite{david:2009}. This random tree ensemble is very different in nature from the simply generated trees in the sense that it arises from a growth process with non-local action which results in a non-equilibrium statistical mechanics model. An effect similar to the emergence of a unique infinite spine is also observed in triangulation models in quantum gravity where a unique large ``universe'' emerges with finite baby--universes attached, see e.g.~\cite{Angel:2002ta}.


A simple observable of a random tree ensemble is its \emph{Hausdorff or fractal dimension} defined as $d_h$ provided that the number of edges within a distance $R$ from the root, denoted by $|\C{B}_R|$, scales as 
\beq
|\C{B}_R|\approx R^{d_h} \quad \mathrm{as}~ R \rightarrow \infty.
\eeq
The meaning of ``$\approx$`` is that both sides are asymptotically the same in a sense which will be made precise in Section \ref{s:spectraldim}.  Another notion of dimensionality of the random tree ensemble is the so-called \emph{spectral dimension}. Let $p(t)$ be the probability that a simple random walk on a tree which starts at the root at time 0 is back at the root at time $t$. If
\beq
p(t) \approx t^{-d_s/2} \quad \mathrm{as}~ t \rightarrow \infty
\eeq 
we say that the spectral dimension of the tree is $d_s$.  The spectral dimension of various classes of graphs has been studied by probabilists \cite{coulhon:2005,Barlow2006,Croydon:2008,fujii:2008,barlow:2009} and by physicists, especially in the quantum gravity literature, see e.g.~\cite{Ambjorn:1997jf,Jonsson:1997gk,Ambjorn:2005db,durhuus:2009b}. Coulhon \cite[Theorem 7.7]{Coulhon2000} derived that one has 
\beq \label{eq:Cou}
\frac{2\dhh}{1+\dhh}\leq\ds\leq \dhh
\eeq
for fixed graphs under certain regularity assumptions. As we will see in the following, random tree ensembles with an infinite spine are tight with the left-hand-side of the inequality and we show below how one can obtain an improved inequality for those tree ensembles which is tight on both sides for many examples.
 
 
We develop methods which relate the diffusion of the random walk to properties of the volume distribution of the finite outgrowths. An important concept is the \emph{hull dimension} defined as $\bar{d}_h$, provided that the total volume of the \emph{hull} $\bar{\C{B}}_R$, the outgrowths from vertices on the spine within distance $R$ from the root, scales as 
\beq
|\bar{\C{B}}_R|\approx R^{\bar{d}_h} \quad \mathrm{as}~ R \rightarrow \infty 
\eeq
see Figure \ref{f:ball}. 

Using simple generating function techniques we begin by proving in Theorem \ref{thmspectral} that the following bounds hold for trees with a unique infinite spine
\beq \label{eq:maininequality}
\frac{2\dhh}{1+\dhh}\leq\ds\leq\frac{2\dhhb}{1+\dhhb} 
\eeq
which, in many examples of interest, tightens the general bound \eqref{eq:Cou} for fixed graphs. Intuitively, if the probability that the outgrowths reach height $n$ decays fast enough with $n$, $d_h$ and $\bar{d}_h$ should be close in which case (\ref{eq:maininequality}) gives a tight bound on $d_s$. We make this a precise statement in Theorem \ref{thm:dh} in the case when the outgrowths from different vertices on the spine are independent and identically distributed ($i.i.d.$). There, we prove that if the random variable $|\C{A}^i|$, denoting the size of the  outgrowths from a vertex on the spine, has a probability generating function
 $\langle z^{|\C{A}^i|}\rangle = 1-(1-z)^{\alpha} l(1-z),$
where $l(x)$ varies slowly at zero, then almost surely
\beq\label{resII}
\dhh\leq\frac{1}{\alpha},\quad \quad \ds\leq\frac{2}{1+\alpha}.
\eeq
Furthermore, we show that if the probability that the height of the outgrowths from a vertex on the spine exceeds $n$ falls off as the inverse of $n$ then one has equality in \eqref{resII}.


Up to that point, our results are very general and we devote the rest of the article to applying the results to specific models. In Section \ref{s:nongeneric} we calculate the spectral dimension of the so--called non--generic, critical phase of simply generated trees, giving an alternative proof to the recent results by Croydon and Kumagai \cite{Croydon:2008}. We note here that our result is weaker in the sense that we define $d_s$ through the singular behaviour of the generating function of $p(t)$, rather than from the asymptotic behaviour of $p(t)$. In Section \ref{s:ag} we calculate the Hausdorff and spectral dimension of the attachment and grafting model for a certain parameter range, extending results in \cite{stefansson:2011a}. We leave the generalisation to the full range of parameters as an open problem.

\section{Hausdorff and spectral dimension of trees with infinite spine}  \label{s:spectraldim}

\subsection{Trees with infinite spine and their ball and hull dimension}

Let us denote by $\Gamma$ the set of all planar rooted locally finite trees, where the root has valency one.\footnote{Note that the assumption that the root has valency one is for conventional reasons and the analysis which follows could in principle also been done without it.} The planarity condition simply means that edges containing the same vertex are ordered around that vertex. One has $\Gamma=\Gamma'\cup\Gamma^\infty$, where $\Gamma'$ is the set of finite such trees and  $\Gamma^\infty$ the set of infinite such trees. We often denote elements of $\Gamma'$ by $T,T_1$ etc., while elements of $\Gamma^\infty$ are called $\tau,\tau_1$ etc. For $T\in\Gamma'$, let $|T|$ be the size of the tree, i.e.~the number of edges in $T$. Let $r$ denote the root of the tree. We consider in this article trees which have a unique, non--backtracking path from the root to infinity and call this path {\it the spine}. Denote the set of such trees by $\Gamma_S^\infty\subset\Gamma^\infty$. Furthermore, denote the vertices on the spine, ordered away from the root $r$, by $s_1,s_2,s_3,...$. We also denote the valency of a vertex $v$ by $\sigma(v)$, i.e. for the root $r$ one has $\sigma(r)=1$. 

For a given $\tau\in\Gamma$ we denote by $\C{B}_R\subset \tau$ the \emph{ball} of radius $R$ around the root, this is the subgraph of $\tau$ spanned by the vertices which have a graph distance from the root which is less than or equal to $R$. Another important concept in the article is the notion of the \emph{hull} $\bar{\C{B}}_R$. For a  $\tau\in\Gamma$ the hull $\bar{\C{B}}_R$ denotes the subgraph which is the union of the spine from $r$ up to vertex $s_R$ and all finite trees attached to it. The concepts of ball and hull are illustrated in Figure \ref{f:ball}; it is easy to see that $\C{B}_R\subseteq\bar{\C{B}}_R$.

A random tree ensemble $(\Gamma,\nu)$ is a set of trees $\Gamma$ equipped with a probability measure $\nu$. An important notion of dimensionality of a random tree ensemble, as discussed in the introduction, is the so-called Hausdorff or fractal dimension, which describes the growth of the size of a ball of radius $R$ and can formally be defined as
\beq \label{defdhh}
\dhh=\lim_{R\to\infty} \frac{\log |\C{B}_R|}{\log R}.
\eeq
Throughout the paper we consider the following strong criterion for existence of $d_h$: There exists an $R_0>0$ such that for $R>R_0$
\beq\label{existencedh}
R^\dhh \underline{L}(R)<|\C{B}_R|< R^\dhh \bar{L}(R),
\eeq
where $ \underline{L}(R), \bar{L}(R)$ are functions which vary slowly at infinity. A real--valued, positive, measurable function $L$ is said to vary slowly at infinity if for each $\lambda > 0$
\begin {equation}
 \lim_{R\rightarrow\infty} \frac{L(\lambda R)}{L(R)} = 1.
\end {equation}
We refer to \cite{seneta:1976} for some properties of slowly varying functions and in particular the fact that (\ref{existencedh}) implies (\ref{defdhh}).

If for a random tree ensemble $(\Gamma,\nu)$ one has that \eqref{defdhh} holds for $\nu$-almost all $\tau\in\Gamma$ then we call $\dhh$ the \emph{quenched Hausdorff dimension}. Furthermore, one can also define the \emph{annealed Hausdorff dimension} through the expected size of the ball of radius $R$
\beq \label{defdH}
\dH=\lim_{R\to\infty} \frac{\log \avg{|\C{B}_R|}}{\log R}.
\eeq
Furthermore, we say that the quenched Hausdorff dimension $\dhh$ exists in the sense of \eqref{existencedh} if the latter is fulfilled for $\nu$-almost all $\tau\in\Gamma$. 

In analogy, we also define the \emph{quenched hull dimension}, if for $\nu$-almost all $\tau\in\Gamma$ 
 \beq \label{defdhhb}
\dhhb=\lim_{R\to\infty} \frac{\log |\bar{\C{B}}_R|}{\log R}
\eeq
where existence is again provided through \eqref{existencedh} analogous to the Hausdorff dimension.
From the fact that $\C{B}_R\subseteq\bar{\C{B}}_R$ it is then also clear that one has $\dhh\leq\dhhb$.

\begin{figure}[t]
\centerline{\scalebox{0.32}{\includegraphics{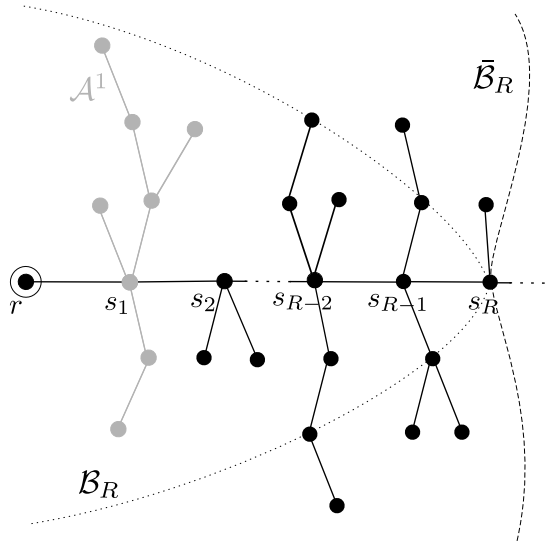}}}
\caption{An illustration of the ball and the hull in a tree with a unique infinite spine started at the root $r$ and labeled by vertices $s_1,s_2, ...$ away from the root. The ball $\C{B}_R$ is the subgraph of all vertices within distance $R$ from the root. The hull $\bar{\C{{B}}}_R$ is composed of the spine up to distance $R$ from the root and the finite outgrowths from its vertices $s_1,...,s_R$. The finite outgrowths from vertex $s_i$ are denoted by $\C{A}^i$.}\label{f:ball}
\end{figure}


\subsection{Random walks on trees and spectral dimension}

In the following, we consider simple, discrete time random walks on trees i.e.~walks which, in each discrete time step, jump to a nearest neighbour with uniform probability. For a given $\tau\in\Gamma$ we denote by $\Omega_\tau$ the set of finite walks on $\tau$. For a walk $\omega\in\Omega_\tau$ we denote by $|\omega|$ the length of the walk and for $t\leq|\omega|$, $\omega(t)$ denotes the vertex where the $\omega$ is located after $t$ steps. We now formulate generating functions for the return probabilities of random walks on a given tree $\tau$.

An important quantity is the generating function for the probability of first return to the root $r$ given by
\beq P_\tau(x)=\sum_{t=0}^\infty\sum_{\omega\in\Omega_\tau: |\omega|=t} \Prob_\tau\left(\{\omega(t)=r\,\vert\, \omega(0)=r, \omega(t')\neq r, 0<t'<t\}\right)(1-x)^{\half t}
\eeq
and for all returns to a starting point which is given by
\beq Q_\tau(x)=\sum_{t=0}^\infty\sum_{\omega\in\Omega_\tau: |\omega|=t} \Prob_\tau\left(\{\omega(t)=r\, \vert \, \omega(0)=r\right\})\,(1-x)^{\half t}.
\eeq
By decomposing a walk which returns to the root into a first return, second return etc., the two generating functions can be related through
\beq \label{QPrel}
Q_\tau(x) = \sum_{n=0}^\infty \left(P_\tau(x)\right)^n = \frac{1}{1-P_\tau(x)}.
\eeq


The generating function $Q_\tau$ encodes all information of the asymptotic of the return probability and one can extract the spectral dimension from it. In particular, for a given $\tau\in\Gamma$ we define the spectral dimension through
\beq \label{Qdefds}
\ds = 2\left( 1+\lim_{x\to0} \frac{\log Q_\tau(x) }{\log x}  \right),
\eeq
provided that $Q_\tau(x)$ diverges as $x\to0$. As for the Hausdorff dimension we also consider the strong criterion for existence, that there exists an $x_0\in(0,1)$ such that for $x<x_0$
\beq \label{existenceds}
 x^{-1+\ds/2} \underline{l}(x) < Q_\tau(x)< x^{-1+\ds/2} \bar{l}(x)
\eeq
where $ \underline{l}(x)$ and $\bar{l}(x)$ are functions which vary slowly at $x\to0^+$.
Notice that the definition of $d_s$ through \eqref{Qdefds} is slightly weaker than the more common definition
\beq \label{Qdefds2}
\ds = -2 \lim_{t\to\infty} \frac{\log \Prob\left(\{\omega(2t)=r\, \vert \, \omega(0)=r\right\}) }{\log t}.
\eeq

For an ensemble $(\Gamma,\nu)$ we say it has \emph{quenched spectral dimension} $d_s$ if \eqref{Qdefds} is fulfilled for $\nu$-almost all trees $\tau\in\Gamma$. Furthermore, we define the \emph{annealed spectral dimension} $d_S$ through
\beq \label{Qdefdsquenched}
\dS = 2\left( 1+\lim_{x\to0} \frac{\log \avg{Q_\tau(x)} }{\log x}  \right),
\eeq
provided that $\avg{Q_\tau(x)}$ diverges as $x\to0$. 


To illustrate the concepts introduced in this subsection consider the following:

\begin{example} The easiest example is the (non-random) tree consisting of a single infinite spine which we denote by $\C{S}$. One has trivially that $|\C{B}_R|=|\bar{\C{B}}_R|=R$ and thus $\dhh=\dhhb=1$. To calculate the first return probability generating function $P_\C{S}(x)$ note that the random walk leaves the root with probability $1$, then at vertex $s_1$ the walker returns to the root with probability $1/2$ or leaves to the right with probability $1/2$ and makes an excursion until its first return to $s_1$. It then returns to the root with probability $1/2$ or does another excursion with probability $1/2$, and so on. Hence, one has for the generating function that
\beq \label{exampleP}
P_\C{S}(x)=\frac{1}{2}  (1-x)\sum_{n=0}^\infty \left( \frac{1}{2} P_\C{S}(x) \right)^n = \frac{1-x}{2-P_\C{S}(x)}.
\eeq
From \eqref{exampleP} it follows that 
\beq
P_\C{S}(x)=1-\sqrt{x}
\eeq
and thus from \eqref{QPrel} one has
\beq
Q_\C{S}(x)=\frac{1}{\sqrt{x}}.
\eeq
Hence, using \eqref{Qdefds} one finds that $\ds=1$ as expected.
\end{example}


\subsection{Relating spectral and Hausdorff dimension}

The main result of this section is the following theorem which we prove in the two subsections below.

\begin{theorem} \label{thmspectral}
For any tree $\tau\in\Gamma_S^\infty$ with a unique infinite spine for which $\dhh$, $\dhhb$ and $\ds$ exist, the following holds
\beq
\frac{2\dhh}{1+\dhh}\leq\ds\leq\frac{2\dhhb}{1+\dhhb}.
\eeq
\end{theorem}

In the proof of this  theorem, in the following two subsections, we use results developed in \cite{durhuus:2007} in the context of generic trees and show that they can be employed in a more general context by introducing the hull dimension.  When  assuming existence of the dimensions we will in practice use the stronger criteria (\ref{existencedh}) and (\ref{existenceds}) although the results also easily follow from assuming existence in the usual sense (\ref{defdhh}) and (\ref{Qdefds}).


\subsubsection{Upper bound on the spectral dimension} 
For any tree $\tau\in\Gamma$ one can easily derive the following recursion relation \cite{durhuus:2007} analogous to \eqref{exampleP} as presented in the example above
\beq\label{Piteration}
P_\tau (x)=\frac{1-x}{ k_\tau+1-\sum_{i=1}^{k_\tau} P_{\tau _i}(x)},
\eeq
where the $\tau_1,...,\tau_{k_\tau}$ denote the $k_\tau$ trees meeting the edge from the root to the unique vertex next to it. In \cite{durhuus:2007} the following simple lemma is proven which follows straightforwardly from \eqref{Piteration} by induction
\begin{lemma} \label{Qfinitetree}
For all finite trees $T\in\Gamma'$ one has
\beq
P_T (x) \geq 1-|T|x.
\eeq
\end{lemma}
If the root of $T$ has degree $k$ the above equation is generalized to $P_T(x) \geq 1- \frac{|T|}{k} x$. To prove the upper bound on the spectral dimension we use the following lemma which is a reformulation of Lemma 7. in \cite{durhuus:2007}
\begin{lemma} \label{lm:lspine}
For all $\tau\in \Gamma^\infty_S$ and $R\geq1$ one has
\beq
P_\tau(x)\;\geq\; 1-\frac{1}{R}  -x |\bar{\C{B}}_R|\,.
\eeq
\end{lemma}

\begin{proof}
The proof follows directly the proof of Lemma 7 in \cite{durhuus:2007}. Denote by $P_\tau^R(x)$ the generating function for the first return probability of walks which do not visit the vertex $s_{R+1}$ on the spine.
An induction proof using the recursion relation \eqref{Piteration} then easily shows that 
\beq
P_\tau^R(x)\;\geq\; 1-\frac{1}{R} - xR -\sum_{i=1}^R(\sigma(s_i)-2)(1-P_{\C{A}^i}(x))\,,
\eeq
where $\C{A}^i$ denotes the union of the vertex $s_i$ and the finite trees attached to it. Noting that $P_\tau(x)\geq P_\tau^R(x)$, using Lemma \ref{Qfinitetree} and observing that $|\bar{\C{B}}_R|=\sum_{i=1}^R |\C{A}^i|+R$ completes the proof.
\end{proof}

To prove the upper bound in the theorem we note that from the existence of the hull dimension one has that $|\bar{\C{B}}_\tau(R)|\leq R^{\dhhb}\bar{L}(R)$, where $\bar{L}(R)$ varies slowly at $R\to\infty$. Thus one has from the proceeding lemma that
\beq
Q_\tau(x)\geq \frac{R \, }{1+ x R^{\dhhb+1}\bar{L}^{}(R)}
\eeq
Choosing $R=\lfloor x^{-1/(1+\dhhb)} \rfloor$ one has
\beq
Q_\tau(x)\geq \frac{  x^{-1/(1+\dhhb)} \,  \underline{l}(x)}{ 1+\underline{l}(x)}, 
\eeq
where $\underline{l}(x)=1/\bar{L}(\lfloor x^{-1/(1+\dhhb)} \rfloor)$ varies slowly at $x\to0^+$. Provided that $d_s$ exists and from the fact that $\dhhb \geq 1$, this proves that $\ds\leq2\dhhb/(1+\dhhb)$.



\subsubsection{Lower bound on the spectral dimension}

To prove the lower bound on the spectral dimension we exploit the following lemma derived in \cite{durhuus:2007} in the context of generic trees and note that it can be used in a more general situation.
\begin{lemma}\label{lmQup}
For all $\tau\in\Gamma$ one has
\beq
Q_\tau(x)\leq R +\frac{2}{x |\C{B}_\tau(R)|}
\eeq
\end{lemma}

The proof of this lemma can be found in \cite{durhuus:2007}. It uses a decomposition of walks in $\Omega_1$, the set of walks which do not reach further than distance $R$ from the root, and $\Omega_2$, the set of walks which do reach further. Then one can show that $Q_\tau(x)=Q^{\Omega_1}_\tau(x)+Q^{\Omega_2}_\tau(x)$, where  $Q^{\Omega_1}_\tau(x)\leq R$ and $Q^{\Omega_2}_\tau(x)\leq \frac{2}{x |\C{B}_\tau(R)|}$.

To prove the lower bound in the theorem we observe that by existence of the  Hausdorff dimension one has that $|\C{B}_\tau(R)|\geq R^{\dhh}\underline{L}(R)$, where $\underline{L}(R)$ varies slowly at $R\to\infty$. Thus, by Lemma \ref{lmQup}
\beq
Q_\tau(x)\leq R +x^{-1} R^{-\dhh} \underline{L}^{-1}(R)
\eeq
Choosing $R=\lceil x^{-1/(1+\dhh)} \rceil$ (which also optimises the inequality) proves that $\ds\geq2\dhh/(1+\dhh)$.



\subsection{Independent and identically distributed outgrowths}

If one considers random tree ensembles with a unique infinite spine, where the outgrowths from different vertices along the spine are independent and identically distributed ($i.i.d.$), one can make stronger statements about the dimensions of the tree ensemble. Recall, the definition of the hull $\bar{\C{B}}_R$ as the union of the spine from $r$ up to vertex $s_R$ and all finite trees attached to it. Let us furthermore denote by $\C{A}^i$ the union of the vertex $s_i$ and the finite trees attached to it. 
Denote by $X^i_n$ the number of vertices in $\C{A}^i$ at a distance $n$ from $s_i$, e.g.\ $X^i_0=1$. Furthermore, denote by $\C{A}^i_n$ the intersection of $\C{A}^i$ with the ball of radius $n$ centered around $s_i$. We have the following Theorem:

\begin{theorem} \label{thm:dh}
Let $(\Gamma,\nu)$ be a random tree ensemble concentrated on the set of trees with a unique infinite spine $\Gamma^\infty_S$ and $i.i.d.$ outgrowths $\left(\C{A}^i\right)_{i\geq 1}$ on the spine.
\begin{enumerate}
\item [(i)] If for $\alpha\in (0,1)$ and $z\in(0,1)$,
\beq\label{thmcond1}
 \avg{z^{|\C{A}^i|}} = 1-(1-z)^{\alpha} l(1-z),
\eeq
 where $l(x)$ varies slowly at $x\to0^+$, then almost surely
\beq
\dhhb\leq\frac{1}{\alpha},\quad \quad\ds\leq\frac{2}{1+\alpha}
\eeq
provided that $\dhhb$ and $d_s$ exist.
\item [(ii)]If furthermore, 
\beq \label{thmcond2}
\measet{X^i_n>0}\leq \frac{c}{n},
\eeq
where $c>0$ is a constant, then almost surely $d_h$ and $d_s$ exist and
\beq
\dhh=\frac{1}{\alpha},\quad \quad \ds=\frac{2}{1+\alpha}.
\eeq
\end{enumerate}
\end{theorem}


\subsubsection{Proof of first part of Theorem \ref{thm:dh}}

 Since $|\bar{\C{B}}_R|=\sum_{i=1}^R |\C{A}^i|+R$ is just a sum of $i.i.d.$ random variables $ |\C{A}^i|$ (with possibly infinite variance), one can easily determine a bound on their sum $|\bar{\C{B}}_R|$ and with that a bound on $\dhhb$ almost surely. We begin by observing that by a Tauberian theorem \cite{Feller}, (\ref{thmcond1}) is equivalent to
\beq\label{thmcond1a}
 \measet{|\C{A}^i|>R} \sim R^{-\alpha} L(R),
\eeq
where $L(R)$ varies slowly at $R\to\infty$. Here, we define ``$\sim$'' to mean that the ratio of the two sides tends to one as $R \rightarrow \infty$. By \eqref{thmcond1a} there exists a function $L_1(R)$ which varies slowly at infinity such that
\beq
\sum_{i=1}^\infty \measet{|\C{A}^i|> R^\frac{1}{\alpha} \, L_1(R) }   <\infty.
\eeq
It then follows from a theorem by Feller \cite[Theorem 2]{feller:1946} that  
\beq
\measet{\tau:\sum_{i=1}^R |\C{A}^i| > R^\frac{1}{\alpha} \, L_1(R)\quad  \mathrm{infinitely \, often}}  = 0
\eeq
One then has that for $\nu$-almost all trees there exists a constant $C>1$ and an $R_0$ such that for $R>R_0$ one has
\beq 
 |\bar{\C{B}}_R| < C R^\frac{1}{\alpha} \, L_1(R).
\eeq
This proves that $\dhhb\leq1/\alpha$ almost surely and using Theorem \ref{thmspectral}, one has further $\ds\leq 2/(1+\alpha)$ almost surely.


\subsubsection{Proof of second part of Theorem \ref{thm:dh}}

In this section we prove that there exists a constant $C$ such that,
\beq\label{part2bound}
\measet{\tau: |\C{B}_R| < \lambda R^\frac{1}{\alpha}}  \leq  C e^{-\lambda^{-\alpha} l(\lambda^{-1}R^{-1/\alpha})}
\eeq
for $\lambda$ small enough such that $\lambda^{-\alpha} l(\lambda^{-1}R^{-1/\alpha}) > c$ with $l(x)$ from \eqref{thmcond1}.
 Assuming that (\ref{part2bound}) holds we proceed as follows. By \cite[Theorem 1.5]{seneta:1976}, for any $\tilde{c}\in\mathbb{R}$, there exists a function $L_2(R)$, slowly varying as $R\rightarrow\infty$, which satisfies
\begin {equation} \label{l2slow}
 \lim_{R\rightarrow\infty}\frac{L_2(R)^{-\alpha}\ell(L_2(R)^{-1}R^{-1/\alpha})}{\log(R^{\tilde{c}})} = 1.
\end {equation}
To see this, let $L^\ast$ be a slowly varying function conjucate to  $L(R) = 1/l(R^{-1/\alpha})$ i.e.
\begin {equation*}
 \lim_{R\rightarrow\infty} L^\ast(R) L(R L^\ast(R)) = 1
\end {equation*}
 and choose $L_2(R) = \left(L^\ast(\frac{R}{\log R^{\tilde{c}}}) /\log(R^{\tilde{c}})\right)^{1/\alpha}$.

%

Therefore, by choosing $\lambda = L_2(R)$ which obeys the requirement of $\lambda$ given above and $\tilde{c}>1$ one finds that
\beq
\sum_{R=1}^\infty \measet{\tau: |\C{B}_R| < R^\frac{1}{\alpha} \, L_2(R)}  <\infty.
\eeq
Using the Borel-Cantelli lemma one then has that 
\beq
\measet{\tau: |\C{B}_R|  < R^\frac{1}{\alpha} \, L_2(R) \quad  \mathrm{infinitely\, often}}  = 0.
\eeq
Thus, for $\nu$-almost all trees there exists an $R_0$ such that for $R>R_0$ one has
\beq 
 |\C{B}_R| > R^\frac{1}{\alpha} \, L_2(R).
\eeq
This proves that $\dhh\geq1/\alpha$ almost surely. Together with the results of the previous section and Theorem \ref{thmspectral} one has thus almost surely
\beq
\dhh=\frac{1}{\alpha},\quad \quad \ds=\frac{2}{1+\alpha}.
\eeq

The basic idea to prove \eqref{part2bound} is the following: if the finite outgrowths along the spine die out fast enough, i.e.~if \eqref{thmcond2} holds, then the ball $\C{B}_R$ and hull $\bar{\C{B}}_R$ are close in size and one can use the latter to estimate $\C{B}_R$. Recall that $\C{A}^i_n$ is the intersection of $\C{A}^i$ with the ball of radius $n$ centered around $s_i$. An essential ingredient to measure how close $\C{B}_R$ and  $\bar{\C{B}}_R$ are in size is provided by the following lemma:

\begin{lemma} \label{lm:boundmomgen}
For $z\leq1$ one has for the probability generating functions
\beq
\avg{z^{|\C{A}^i|}}\leq \avg{z^{|\C{A}^i_n|}} \leq \avg{z^{|\C{A}^i|}}+ \measet{X^i_n>0}
\eeq
\end{lemma}

\begin{proof} The proof is inspired by ideas of \cite{durhuus:2007}. Firstly, note that $|\C{A}^i|\geq |\C{A}^i_n|$, hence the lower bound follows. For easy notation let us denote the event $\C{E}_n=\{X^i_n>0\}$. Then
\bea
\avg{z^{|\C{A}^i_n|} - z^{|\C{A}^i|}}& =&\int_{\C{E}_n} (z^{|\C{A}^i_n|} - z^{|\C{A}^i|}) d\nu+ \int_{\C{E}^c_n} (z^{|\C{A}^i_n|} - z^{|\C{A}^i|}) d\nu\nn\\
&=& \int_{\C{E}_n} (z^{|\C{A}^i_n|} - z^{|\C{A}^i|}) d\nu \leq \int_{\C{E}_n} z^{|\C{A}^i_n|} d\nu \nn\\
&\leq& \int_{\C{E}_n}d\nu= \measet{X^i_n>0}
\eea
Thus the upper bound follows.
\end{proof}

To prove \eqref{part2bound} we first note that since 
\beq
\C{A}^{1}_{[R/2]}\cup \C{A}^{2}_{[R/2]}\cup...\cup \C{A}^{[R/2]}_{[R/2]} \subset \C{B}_R
\eeq
it suffices to prove that 
\beq \label{toproofnew}
\measet{\tau: \sum_{i=1}^R |\C{A}^{i}_R| < \lambda R^\frac{1}{\alpha} }  \leq 
 C e^{-\lambda^{-\alpha} l(\lambda^{-1}R^{-1/\alpha})}.
\eeq
Using Markov's inequality,  the independence of the $|\C{A}^i_R|$ and Lemma \ref{lm:boundmomgen}  we find that for any $t\in(0,1)$ and $\lambda >0$
\bea
\measet{\tau: \sum_{i=1}^R |\C{A}^{i}_R| < \lambda R^\frac{1}{\alpha} } & \leq& \measet{\tau: \exp{(-t\sum_{i=1}^R |\C{A}^{i}_R|)} > \exp{(-t\lambda R^\frac{1}{\alpha})} } \nn\\
&\leq & e^{t\lambda R^\frac{1}{\alpha}} \langle \prod_{i=1}^R e^{-t|\C{A}^{i}_R|} \rangle_\nu = e^{t\lambda R^\frac{1}{\alpha}} \langle  e^{-t|\C{A}^{i}_R|} \rangle_\nu^R \nn\\
&\leq& e^{t\lambda R^\frac{1}{\alpha}} \left( \avg{e^{-t |\C{A}^i| }} +\measet{X^i_R>0} \right)^R. \label{eq:markov}
\eea
From \eqref{thmcond1} it follows that
\beq\label{momentgenA}
\langle e^{-t |\C{A}^i|} \rangle_\nu \leq 1- t^\alpha l(t) +o(t^\alpha l(t))
\eeq
where $l(t)$ varies slowly as $t\to0^+$. 
 Choose $t= \lambda^{-1}R^{-1/\alpha}$ and $\lambda$ small enough such that
\begin {equation} \label{alphacondition}
\lambda^{-\alpha} l(\lambda^{-1}R^{-1/\alpha}) > c
\end {equation}
with $c$ from (\ref{thmcond2}). Apply (\ref{momentgenA}) and (\ref{thmcond2}) to (\ref{eq:markov}) and get, for a suitable constant $C'$
\bea \nn
\measet{\tau: \sum_{i=1}^R |\C{A}^{i}_R| < \lambda R^\frac{1}{\alpha} } & \leq& C' \left(1-\frac{\lambda^{-\alpha} l(\lambda^{-1}R^{-1/\alpha})-c}{R}\right)^R \\
&\leq&C' e^{-\lambda^{-\alpha} l(\lambda^{-1}R^{-1/\alpha})+c}. \label{expbound}
\eea
%
Thus (\ref{part2bound}) follows from (\ref{expbound}) with $C = C' e^c$.



\section{Conditioned Galton-Watson trees} \label{s:nongeneric}

\subsection{The model}

In this section we apply the results of the previous section to the model of simply generated trees. Simply generated trees can, in most cases of interest, be interpreted as a size conditioned  Galton-Watson process. A Galton-Watson process is a discrete time branching process which starts from a single particle. At each time step the particles branch independently to $k$ other particles with the same probability $p(k)$, $k\geq 0$ referred to as {\it offspring probability}. We denote the generating function of the offspring probabilities (or outdegrees) by $f(z) = \sum_{n=0}^\infty p(n) z^n$. The value of $f'(1)$, the mean number of offspring, determines the survival properties of the process. If $f'(1) = 1$ and excluding the trivial case $p(1)=1$ the process is said to be critical and it dies out with probability one. If $f'(1) < 1$ the process is sub--critical and dies out exponentially fast and if $f'(1) > 0$ the process is super--critical and has a positive probability of surviving forever, see e.g.~\cite{Harris:1963}.	

The  model of simply generated trees has as parameters a sequence of non--negative weights $(w_n)_{n\geq 1}$ referred to as {\it branching weights} and is defined by a (Gibbs) measure on the set of trees with $n$ edges  by
\begin {equation}
  \nu_n(T) = Z_n^{-1}\prod_{v\in V(T)\setminus\{r\}} w_{\sigma(v)},\quad T\in \Gamma_n
\end {equation}
 where $Z_n$ is a normalization
\begin {equation}
 Z_n = \sum_{T\in \Gamma_n} \prod_{v\in V(T)\setminus\{r\}} w_{\sigma(v)}
\end {equation}
often referred to as the finite volume partition function. It is useful to define the generating functions
\begin {equation}
 \mathcal{Z}(\zeta) = \sum_{n=1}^\infty Z_n \zeta^n \quad \mathrm{and} \quad g(z) = \sum_{n=0}^\infty w_{n+1}z^{n}
\end {equation}
which are related by the equation
\begin {equation} \label{eq:tree}
 \C{Z}(\zeta) = \zeta g(\C{Z}(\zeta)),
\end {equation}
see e.g.~\cite{flajolet:2009}. We will denote their radii of convergence by $\zeta_0$ and $\rho$ respectively and we furthermore define $\mathcal{Z}_0 = \lim_{\zeta \rightarrow \zeta_0} \mathcal{Z}(\zeta)$. When $\rho > 0$ the measure $\nu_n$ can equivalently be defined by a Galton-Watson branching process with offspring probabilities
\begin {equation} \label{e:offspring}
p(k) = \zeta_0 w_{k+1} \C{Z}_0^{k-1}
\end {equation}
conditioned to have $n$ edges, see e.g.~\cite{jonsson:2011}. By (\ref{eq:tree}), the mean offspring number can be written as $f'(1) = 1-g(\C{Z}_0)/\C{Z}'(\zeta_0)$ and thus it is clear that the process is either critical or sub--critical. Here we will only consider the critical case i.e.~when $\C{Z}'(\zeta_0) = \infty$ in which case the following theorem holds.
\begin {theorem} \label{thm:convergence}
 For a critical size conditioned Galton-Watson process, the finite volume measures $\nu_n$ converge weakly as $n\rightarrow \infty$ towards a measure $\nu$ which is concentrated on $\Gamma_S^\infty$. The degrees of vertices on the spine are independently distributed by
\begin {equation}
 \phi(k) = \zeta_0 (k-1) w_k \C{Z}_0^{k-2}
\end {equation} 
and the outgrowths from the spine are independent Galton-Watson trees with offspring probabilities (\ref{e:offspring}).  
\end {theorem}

The critical model is usually divided into two cases depending on whether $f''(1)$ is finite or infinite. For the case $f''(1) < \infty$, Theorem \ref{thm:convergence} is originally due to Kennedy \cite{Kennedy1975} and later to Aldous and Pitman \cite{Aldous1998}, see also \cite{durhuus:2007}. To our best knowledge, the generalisation which includes the case $f''(1)=\infty$ was first proved in the special case $w_n \sim n^{-\beta}$ by Jonsson and Stef\'ansson \cite{jonsson:2011} and later in full generality by Janson \cite{janson:2011}. The same limiting behaviour was also obtained earlier by Kesten \cite{Kesten1986} for Galton-Watson trees conditioned on their height.

  In the case $f''(1) < \infty$ the trees always belong to the same universality class and have been referred to as {\it generic trees} in the physics literature. In the infinite case there is a range of universality classes depending on the singular behaviour of $f$ and this case has been referred to as {\it critical non-generic}. 
  
 \subsection{Dimensions}
  
 The following theorem holds for the Hausdorff and spectral dimension in the critical case 

\begin{theorem} \label{thmhausdorff}
\
\begin {enumerate}
 \item [(i)] The quenched Hausdorff dimension and spectral dimension of generic trees is almost surely
\begin {equation}
 \dhh=2 \quad \mathrm{and} \quad \ds=4/3
\end {equation}
respectively.
\item [(ii)] For critical non-generic trees with $w_n\sim n^{-\beta}L(n)$, where $\beta\in(2,3]$ and $L$ slowly varying at infinity, the quenched Hausdorff and spectral dimensions are almost surely
\beq
\dhh=\frac{\beta-1}{\beta-2} \quad \mathrm{and} \quad \ds=\frac{2(\beta-1)}{2\beta-3}
\eeq
respectively.
\end {enumerate}

\end{theorem}
The result on $d_h$ in $(i)$ is proven for instance in \cite{durhuus:2009b} and the result on $d_s$ was originally proven in \cite{Barlow2006,fujii:2008}. The results in part $(ii)$ were  conjectured in the physics literature \cite{
correia:1997gf,burda:2001dx} and later proven in the mathematical literature by Croydon and Kumagai \cite{Croydon:2008}. Below we will show how Theorem \ref{thmhausdorff} easily follows from Theorem  \ref{thm:dh}. While the proof in \cite{Croydon:2008} is slightly more general, as discussed in the introduction, the proof presented here is more intuitive and hopefully makes the result more accessible to physicists.

\begin{proof}
It is a standard result, see e.g.~\cite{Athreya1972}, that for generic trees
\begin {equation}
 \langle z^{|\C{A}_i|}\rangle= 1-c(1-\zeta)^{1/2} + o((1-\zeta)^{1/2})
\end {equation}
and
\beq\label{fnz}
\frac{1}{1-f_n(z)}-\frac{1}{1-z}= \frac{1}{2}f''(1) n +o(n),
\eeq
where $f_n=f\circ \cdots \circ f$, $n$-times, is the $n$-th iterate of the generating function $f(z)$ of the offspring distribution. Denoting by $X_n^{i,j}$ the size of the $n$-th generation of the $j$-th outgrowth from $s_i$, one thus has using \eqref{fnz} that for $n\geq 1$
\begin {equation} \label{eq:GWheight}
\measet{X^{i,j}_n>0}= 1-f_{n-1}(0)=\frac{2}{f''(1)n} (1+ o(1)),
\end {equation}
 Hence, since the outgrowths are independent, one has using Theorem \ref{thm:convergence}
 \bea
 \measet{X^{i}_n>0}&=& 1- f'(1- \measet{X^{i,j}_n>0})\nn\\
 &=&  \frac{2}{ n} (1+ o(1)).
 \eea
 The results then follow from Theorem \ref{thm:dh}. 

Next, consider the case $f''(1) = \infty$. By a Tauberian theorem \cite{Feller}, $w_n\sim n^{-\beta}L(n)$, implies that one can write
\begin {equation} \label{eq:fslow}
 f(z) = z+(1-z)^{\beta-1}l_1(1-z)
\end {equation}
where $l_1$ is slowly varying at zero. Let $\C{W}(\zeta) = \C{Z}(\zeta_0\zeta)/\C{Z}_0$. Using generating function arguments one can deduce from Theorem \ref{thm:convergence} that
\begin {equation} \label{eq:EAi}
 \langle z^{|\C{A}_i|}\rangle = f'(\C{W}(z)).
\end {equation}
Write
\begin {equation} \label{eq:Wslow}
 \C{W}(\zeta) = 1-(1-\zeta)^{\frac{1}{\beta-1}} \chi(1-\zeta).
\end {equation}
By (\ref{eq:tree}) and (\ref{eq:fslow}) we find that $\chi$ has to satisfy
\begin {equation} 
 \lim_{x\rightarrow 0} \chi(x)^{\beta-1} l_1\left(x^{\frac{1}{\beta-1}} \chi(x)\right) = 1
\end {equation}
which, by \cite[Theorem 1.5]{seneta:1976}, entails that $\chi$ is slowly varying at zero. Therefore, by (\ref{eq:fslow}),  (\ref{eq:EAi}),  (\ref{eq:Wslow}) and Tauberian theorems one can write
\begin {equation}
 \langle z^{|\C{A}_i|}\rangle = 1-(1-z)^{\frac{\beta-2}{\beta-1}}l_2(1-z)
\end {equation}
with $l_2$ slowly varying at zero. Thus, by part $(i)$ of Theorem \ref{thm:dh} we have established the upper bounds on $d_h$ and $d_s$. 

For the lower bounds we use part $(ii)$ and consider the survival probability of the outgrowths. Let $X_n^{i,j}$ be the size of the $n$-th generation of the $j$-th outgrowth from $s_i$. It was shown by Slack \cite[Lemma 2]{Slack:1968} that 
\beq
 \left(\measet{X^{i,j}_n>0}\right)^{\beta-2}\, l_1(\measet{X^{i,j}_n>0}) =\frac{1}{(\beta-2)n} (1+o(1)),
\eeq
where $l_1$ is the same slowly varying function as in \eqref{eq:fslow}. Then, since the outgrowths are independent, we find by (\ref{eq:fslow}) that
\begin {eqnarray} 
 \measet{X^i_n>0} &=& 1-f'\left(1-\measet{X^{i,j}_n>0}\right) \nn\\
 &=& (\beta-1)(\measet{X^{i,j}_n>0})^{\beta-2}\, l_1(\measet{X^{i,j}_n>0}) (1+o(1)) \nn\\
 &=& \frac{\beta-1}{\beta-2} \,\frac{1}{n} (1+o(1)).
\end {eqnarray}
which completes the proof.
\end{proof}


\section{The attachment and grafting model} \label{s:ag}
\subsection{The model}
The attachment and grafting (ag) model \cite{stefansson:2011a} is a recent model of randomly growing rooted planar trees which is a special case of a very general tree growth model, referred to as the vertex splitting model. The vertex splitting model was introduced in \cite{david:2009} as a modification of a combinatorial model encountered in the theory of random RNA folding \cite{david:2008}.  The original motivation for studying the special case of the ag--model is that it has a so--called \emph{Markov branching property} which makes it exactly solvable in a strong sense.  It furthermore turns out to have a unique infinite spine which enables us to apply Theorem \ref{thm:dh} to study its Hausdorff and spectral dimension. Using the first part of the theorem we establish what we believe to be tight upper bounds on the dimensions for the full range of parameters. The corresponding lower bounds require information on the extinction probability of the outgrowths and we will provide results on that only for a certain range of parameters.

We give an informal description of the growth rules of the ag--model below but refer to \cite{stefansson:2011a} for a more detailed discussion. The model has two parameters $\alpha,\gamma \in [0,1]$ and a constraint $D$ on the maximum degree of vertices. For easier notation we define
\begin {equation} \label {aandb}
 \eta = \left\{ \begin {array}{cc}
 -\frac{1-\alpha}{D-2}  &  ~ \mathrm{if}~ D< \infty, \\
 1-\gamma  & ~ \mathrm{if}~ D= \infty. \\
 	\end {array}\right. \quad 
\end {equation}
Call the edges which are adjacent to vertices of degree one (besides the root) {\it leaves} and call the other edges {\it internal edges}. Starting from the unique tree with two edges, in each time step the number of edges is increased by one by randomly selecting
\begin {list}{\labelitemi}{\leftmargin=2em}
 \item [(a)] a vertex of degree $k\geq 2$ with relative probability $ \eta k + 1 - 2\eta - \alpha $ and attaching a new edge to it (the $k$ possibilities of attaching chosen uniformly at random) or 
\item [(b)]  an inner edge with relative probability $\alpha$ and dividing it into two edges by grafting a vertex to it or 
\item[(c)] a leaf with relative probability $1-\eta$ and dividing it into two edges by grafting a vertex to it,
\end {list}
see Fig.~\ref{f:growth}.
%
%
\begin{figure} [t]
\centerline{\scalebox{0.33}{\includegraphics{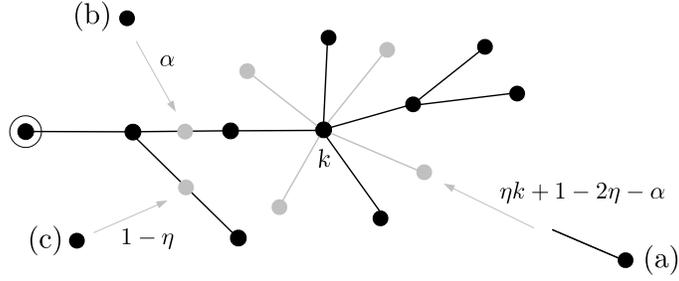}}}
\caption{Growth rules of the ag--model. The root is indicated by a circled vertex.} \label{f:growth}
\end{figure}
%
%
The growth rule generates a sequence of probability measures on $\Gamma$ which we denote by $(\nu_n)_{n\geq 1}$ (not writing explicitly the dependency on $\alpha$, $\gamma$ and $D$). The probability of a tree $T \in \Gamma_{n+1}$ is given by a recursion 
\begin {equation}
\nu_{n+1}(T) = \sum_{T' \in \Gamma_n} \nu_{n}(T') \mathbb{P}(T'\rightarrow T)                                                                                                                                                                   \end {equation}
 where $\mathbb{P}(T'\rightarrow T)$ is the probability of growing $T$ from $T'$ according to the above growth rule.

 
\subsection{Markov branching and convergence of the finite volume measures}
The ag--model has a property called Markov branching, a concept originally introduced by Aldous \cite{ald:1996}. Markov branching means that for any $k\geq 2$ there is a  function $q_k: \mathbb{N}^{k-1} \rightarrow \mathbb{R}_+$ such that for each finite tree $T_0$ which branches at the nearest neighbour of the root to subtrees $T_1,\ldots,T_{k-1}$ it holds that
\begin {equation} \label{e:mb}
 \nu_{|T_0|}(T_0) = q_k(|T_1|,\ldots,|T_{k-1}|) \prod_{i=1}^{k-1} \nu_{|T_i|}(T_i).
\end {equation}
The functions $q_k$ are referred to as {\it the first split distribution}. In the ag--model, the first split distribution is given by \cite{stefansson:2011a}
\begin {eqnarray} \nonumber
 q_k(n_1,\ldots,n_{k-1}) &=& \frac{ \Gamma\left(k-2+\frac{1-\alpha}{\eta}\right)}{\Gamma\left(\frac{1-\alpha}{\eta}\right)\Gamma(k)} \frac{\G{1-\eta}\G{n}}{\eta\Gamma\left(n-\eta\right)}\prod_{i=1}^{k-1} \frac{\eta\Gamma(n_i-\eta)}{\G{1-\eta}\G{n_i+1}}  \\ 
&& \times ~\left(1-\eta - \alpha + \alpha \sum_{i=1}^{k-1}\frac{n_i}{n-n_i}\right)  \label{qfuncgen}
\end {eqnarray}
where $n_1 + \cdots + n_{k-1} = n-1$.

In \cite{stefansson:2011a} it was shown, using (\ref{e:mb}) and (\ref{qfuncgen}), that the sequence $(\nu_n)_{n\geq 1}$ converges weakly as $n \rightarrow \infty$ to a measure $\nu$ which is concentrated on the set of infinite trees with exactly one infinite spine having finite outgrowths. Outgrowths from different vertices on the spine are independent and the probability that a vertex on the spine has degree $k$ and that the finite outgrowths attached to it are $T_1,\ldots,T_{k-2}$, is given by
\begin {equation} \label{muinfgen}
 \mu_k(T_1,\ldots,T_{k-2})\! =\!  \frac{\alpha \Gamma\! \left(k-2+\frac{1-\alpha}{\eta}\right)}{\G{\frac{1-\alpha}{\eta}}\Gamma(k-1)(1+m)}\prod_{i=1}^{k-2} \frac{\eta \Gamma(|T_i|-\eta)\nu_{|T_i|}(T_i)}{\G{1-\eta}\Gamma(|T_i|+1)}, 
\end {equation}
where $m = |T_1| + \cdots + |T_{k-2}|$ (with the convention that $\mu_2(\emptyset) = 1$).


\subsection{Dimensions}
The annealed Hausdorff dimension of the ag--model was calculated in \cite{stefansson:2011a} using a generating function argument. It was shown that $\dH = 1/\alpha$ when $D$ is finite and in the case $D=\infty$ and $\alpha > 1-\gamma$. Note that it is meaningless to consider the annealed Hausdorff dimension when $D= \infty$ and $\alpha \leq 1-\gamma$ since then it follows from (\ref{muinfgen}) that the expected degree of a vertex on the spine is infinite and thus $\dH = \infty$.

Despite this, one might still conjecture that a.s.~$\dhh = 1/\alpha$. This was confirmed in \cite{stefansson:2011a} when $\gamma = 0$ in which case the outgrowths from the spine are single edges and it was furthermore shown that in this case $\ds = 2 \dhh/(1+\dhh) = 2/(1+\alpha)$. Below, we extend these results to a wider range of parameters.
\begin {theorem} \label{thm:agdim}
\begin {enumerate}
 \item For the ag--ensemble $(\Gamma,\nu)$ it holds for any $\alpha$, $\gamma$ and $D$,  that almost surely
\begin {equation} \label{e:agupperbound}
 \dhh \leq \frac{1}{\alpha} \qquad \mathrm{and} \qquad \ds \leq \frac{2}{1+\alpha}. 
\end {equation}
\item Furthermore, for $D=3$, $\frac{1}{2} \leq \alpha \leq 1$ and for $D=\infty$, $\gamma = 0$, we have almost surely
\begin {equation} \label{e:agequal}
 \dhh = \frac{1}{\alpha} \qquad \mathrm{and} \qquad \ds = \frac{2}{1+\alpha}. 
\end {equation}
\end {enumerate}
\end {theorem}

\begin {proof}
As before, we denote the outgrowths from vertex $s_i$ on the spine by $\C{A}^i$. The upper bounds in (\ref{e:agupperbound}) follow from the simple observation that
\begin {eqnarray} \nonumber
\langle z^{|\C{A}^i|}\rangle_{\nu} &=& \sum_{k=2}^\infty \sum_{m=k-2}^\infty \sum_{n_1+\cdots+n_{k-2} = m} \sum_{T_1 \in \Gamma_{n_1},\ldots,T_{k-2}\in\Gamma_{n_{k-2}}} \!\!\!\!\!\!\!\!\! \mu_k(T_1,\ldots,T_{k-2}) z^m \\
&=&\frac{1}{z}\int_0^z dx \sum_{k=2}^\infty \frac{\alpha \Gamma\left(k-2+\frac{1-\alpha}{\eta}\right)}{\G{\frac{1-\alpha}{\eta}}\Gamma(k-1)} \left[ 1-(1-x)^\eta \right]^{k-2}\nn\\
&=& \frac{1-(1-z)^\alpha}{z}. \label{obs}
\end {eqnarray}
The second equality in (\ref{obs}) is obtained by noting that the innermost sums of $\nu_{n_i}(T_i)$ over $T_i$, $i=1,\ldots,k-2$, yield 1 and then using standard gamma function identities \cite{abramowitz}
\beq \label{gammaidentity}
\sum_{n=1}^\infty \frac{\sigma \Gamma(n-\sigma)}{\Gamma(1-\sigma)\Gamma(n+1)}z^n = 1-(1-z)^\sigma.
\eeq
Finally, an application of the first part of Theorem \ref{thm:dh} yields (\ref{e:agupperbound}).

To obtain the equalities in (\ref{e:agequal}) we study the height distribution of the outgrowths from the spine and apply the second part of Theorem \ref{thm:dh}. We start by noting that the case when $D=\infty$ and $\gamma = 0$, which was already proved in \cite{stefansson:2011a}, follows immediately from Theorem \ref{thm:dh} since in this case the outgrowths from the spine all have height 1. We now focus on the case $D=3$. As before, denote by $X^i_n$ the size of the $n$--th generation of $\C{A}^i$. The result follows from 
\begin {lemma} \label{l:agh}
 When $D=3$ and $\frac{1}{2}  \leq \alpha \leq 1$ it holds that
\begin  {equation}
 \measet{X^i_n>0}\leq \frac{2}{n}(1+o(1)).
\end  {equation}
\end {lemma}
We devote the following subsection to the proof of this lemma.
\end {proof}


\subsection{Proof of Lemma \ref{l:agh}}
We start the proof with a general approach and reduce to the specific parameters stated in Lemma \ref{l:agh} only in the end when needed. In this way we keep the problem of extending the results to all parameters clearly approachable. 

For a finite tree $T$, let $h(T)$ be its height, i.e.~the maximum distance from the root to any vertex. We make the following definitions for more compact notation
\begin {equation}
 c_k = \frac{\G{k-2+\frac{1-\alpha}{\eta}}}{\G{\frac{1-\alpha}{\eta}}\G{k}} \quad\mathrm{and}\quad
 C_{n,R} =  \frac{-\eta \Gamma(n-\eta)\nu_n(\{h(T) \leq R\})}{\Gamma(1-\eta)\Gamma(n+1)}
\end {equation}
with the convention that $\nu_0(\{h(T) \leq R\}) = 1$. The main tool that we will use in the proof is the following generating function
\begin {equation} \label{ag:Hdef}
 H_R(\zeta) = \sum_{n=0}^\infty C_{n,R} \zeta^n.
\end {equation}

\begin {proposition} \label{prop:agfinal}
The probability that $\C{A}^i$ is extinct at level $R+1$ is given by
\begin {equation} \label{e:agfinal}
 \nu(\{X_{R+1}^i=0\}) = \alpha \int_0^1 (H_R(\zeta))^{\frac{\alpha-1}{\eta}} d\zeta. 
\end {equation}
\end {proposition}

\begin {proof}
 Using the distribution of the outgrowths given in (\ref{muinfgen}) one can write
\begin {eqnarray} \nonumber
 && \nu(\{X_{R+1}^i=0\}) = \nu(\{h(\C{A}^i)\leq R\}) \\
&&=~\alpha\int_0^1\Bigg[\sum_{k=2}^\infty (k-1)c_k \sum_{m=k-2}^\infty  \sum_{n_1+\cdots+n_{k-2} = m, n_i \geq 1} \prod_{i=1}^{k-2} (-C_{n_i,R})\zeta^{n_i}\Bigg]  d\zeta. 
\end {eqnarray}
The sum over $m$ is first performed giving $(1-H_R(\zeta))^{k-2}$ and the sum over $k$ is then performed using standard gamma function identities similar as above giving (\ref{e:agfinal}).
\end {proof}
Proposition \ref{prop:agfinal} shows that $H_R(\zeta)$ contains all the information needed to calculate the extinction probability of $\C{A}^i$. Next, using the Markov branching property, we will derive a recursion equation for $H_R(\zeta)$ which enables us to extract enough information.

\begin {proposition}
For $R\geq 1$, it holds that
\bea \label{ag:hrec}
 H_R(\zeta)& =& 1+\alpha H_{R-1}(\zeta) \int_0^\zeta (H_{R-1}(y))^{\frac{\alpha-1}{\eta}}dy
 - (\alpha+\eta)\int_0^\zeta (H_{R-1}(y))^{\frac{\alpha+\eta-1}{\eta}}dy
\eea
with the convention that $H_0(\zeta) = 1$.
\end {proposition}

\begin {proof}
Using the Markov branching property and (\ref{qfuncgen}) we write
\begin {eqnarray}
-n C_{n,R}  &=& \sum_{k=2}^\infty c_k \sum_{n_1+\cdots+n_{k-1} = n-1, n_i\geq 1} \prod_{i=1}^{k-1} (-C_{n_i,R-1} )\nn\\
&&\times~\left(1-\eta-\alpha + \alpha \sum_{j=1}^{k-1}\frac{n_j}{n-n_j}\right).
\end {eqnarray}
Multiply by $\zeta^{n-1}$ and sum from $n=2,\ldots\infty$ to get
\begin {eqnarray}
 &&-(H_R'(\zeta)+\eta) = \sum_{k=2}^\infty c_k \sum_{n=k}^\infty \sum_{n_1+\cdots+n_{k-1} = n-1, n_i\geq 1} \prod_{i=1}^{k-1} (-C_{n_i,R-1}) \zeta^{n_i}\nn \\
&&\times~\left(1-\eta-\alpha + \alpha \sum_{j=1}^{k-1}\frac{n_j}{n-n_j}\right) \nn\\
&& = (1-\eta-\alpha) \sum_{k=2}^\infty c_k \left(-\sum_{n=1}^\infty C_{n,R-1}\zeta^n\right)^{k-1}\nn \\
&& -~\alpha\left(\frac{d}{d\zeta}\sum_{n=1}^\infty C_{n,R-1}\zeta^n\right)\int_{0}^\zeta\sum_{k=2}^\infty (k-1)c_k \left(-\sum_{n=1}^\infty C_{n,R-1}y^n\right)^{k-2} dy.
\end {eqnarray}
Using the definition of $H_R$ and performing the sum over $k$ yields
\begin {eqnarray}
 &&H_R'(\zeta) =-\eta \left(H_{R-1}(\zeta)\right)^{\frac{\alpha+\eta-1}{\eta}} + \alpha H'_{R-1}(\zeta)\int_{0}^\zeta \left(H_{R-1}(y)\right)^{\frac{\alpha-1}{\eta}}dy.
\end {eqnarray}
Integrating this equation and using integration by parts on the second term on the right hand side finally yields the result (\ref{ag:hrec}).
\end {proof}
We now reduce to the case $D=3$ in which case $\eta = \alpha -1 \leq 0$. The recursion (\ref{ag:hrec}) is then greatly simplified. We define $F_R(\zeta) = \alpha \int_0^\zeta H_R(y)dy$ and then $\nu(\{X_{R+1}^i=0\}) = F_R(1)$. Next we integrate (\ref{ag:hrec}) to get
\begin {equation} \label{ag:Frec}
 F_R(\zeta) = \alpha \zeta + \frac{1}{2} F_{R-1}(\zeta)^2-\frac{2\alpha-1}{\alpha}\int_0^{\zeta}\int_0^{y} \left(F_{R-1}'(x)\right)^2 dx dy.
\end {equation}
One can see directly from the definition of $H_R$ (\ref{ag:Hdef}) that when $\eta \leq 0$, $F_R$ and all its derivatives are increasing in $R$. One can furthermore deduce from (\ref{ag:Hdef}) that 
\begin {equation}
F_\infty(\zeta):=\lim_{R\rightarrow\infty}F_R(\zeta) = 1-(1-\zeta)^\alpha. 
\end{equation}
Therefore, under the assumption that $\frac{1}{2} \leq \alpha \leq 1$ we may insert $F_\infty(\zeta)$
 into the integral in (\ref{ag:Frec}) to get the inequality
\begin {equation}
F_R(\zeta) \geq \frac{1}{2}+\frac{1}{2}\left(F_{R-1}(\zeta)\right)^2-\frac{1}{2}(1-\zeta)^{2\alpha}. 
\end {equation}
Evaluated at $\zeta = 1$ this can be written as
\begin {equation} \label{ag:FfF}
 F_R(1) \geq f(F_{R-1}(1))
\end {equation}
with $f(x) = \frac{1}{2}+\frac{1}{2}x^2.$  We can interpret $f$ as a generating function of the offspring probabilities $p(0) = p(2) = 1/2$ of a critical Galton-Watson process, cf. Section \ref{s:nongeneric}. Let $f_R$ be the $R$--th iterate of $f$ i.e. $f_R = f\circ \cdots \circ f$, $R$ times. Now, $f$ is an increasing function and therefore, by repeatedly applying (\ref{ag:FfF}) one gets
\begin {equation}
 \nu(\{X_{R+1}^i=0\}) = F_R(1) \geq f_{R}(F_0(1)) \geq f_{R}(1/2)
\end {equation}
where in the last step we used that $F_0(1) = \alpha \geq 1/2$. It then follows from \eqref{fnz} that $f_{R}(1/2) = 1-\frac{2}{R}(1+o(1))$. This concludes the proof of Lemma \ref{l:agh}.

\section{Discussion}

We developed relatively simple methods, relying on generating function arguments, for calculating the Hausdorff and spectral dimension of trees with a unique infinite spine resulting in Theorems \ref{thmspectral} and \ref{thm:dh}, where the latter applies when the outgrowths along the spine are independent and identically distributed. These methods were applied to two models of random trees of very different nature, demonstrating the versatility of our results.

The first application of Theorem \ref{thm:dh} concerns non-generic critical trees which are a special case of simply generated trees. The values of the Hausdorff and spectral dimension of the non--generic, critical, size conditioned Galton--Watson trees were conjectured by mathematical physicists 15 years ago \cite{correia:1997gf,burda:2001dx} but only recently proved by mathematicians \cite{Croydon:2008}. We used the opportunity here to communicate these results to the physics community as well as providing a simple alternative proof. 

We would like to point out that these results complete the study of the dimensionality of the different phases of simply generated trees. Theorems \ref{thm:convergence} and \ref{thmhausdorff} concern critical Galton--Watson processes. As was mentioned in Section \ref{s:nongeneric}, the simply generated trees can also correspond to sub--critical Galton--Watson processes. In the sub--critical case, the finite volume measures converge weakly to a measure concentrated on the set of trees with a finite spine ending with a vertex of infinite degree. The length of the spine has a geometric distribution and the outgrowths from the spine are finite, independent, sub--critical Galton--Watson processes. This convergence theorem was originally proved in \cite{jonsson:2011} in the case $w_n \sim n^{-\beta}$ and later in full generality in \cite{janson:2011}. 

Due to the presence of a vertex of infinite degree it clearly follows that almost surely $d_s = \infty$ since a random walker will eventually hit the infinite degree vertex and similarly $d_h = d_H = \infty$. However, it was shown in \cite{jonsson:2011} that when $w_n\sim n^{-\beta}$, the annealed spectral dimension is finite and $d_S = 2(\beta-1)$. This is due to the fact that the fluctuations of the outgrowths from the spine serve to slow the random walker down on the way to the infinite degree vertex. Note that in this case it does not hold that $d_S = 2d_H/(1+d_H)$ as in the other phases due to the absence of the infinite spine.

We summarize the above discussion in the phase diagram in Figure \ref{f:phase} where we consider the case when $w_1$ and $\beta$ are the free parameters and $w_n\sim n^{-\beta}$. This choice of parameters allows us to access all phases and explore the full range of dimensions in each phase, see \cite{jonsson:2011} for a more detailed explanation.

\begin{figure} [t]
\centerline{\scalebox{0.5}{\includegraphics{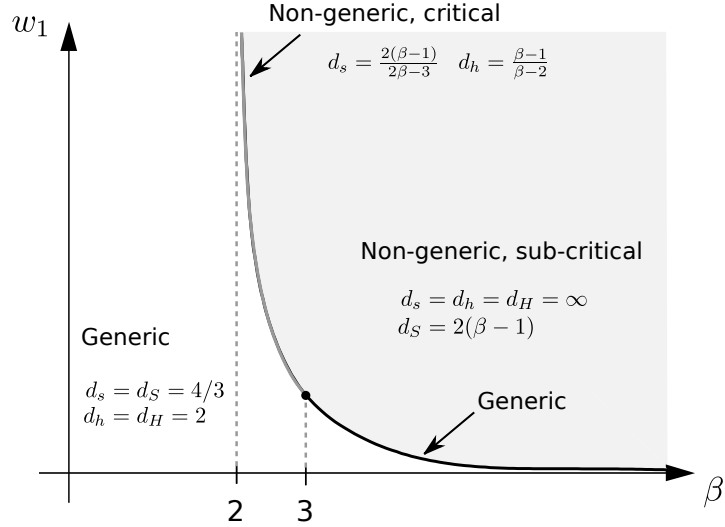}}}
\caption{A phase diagram for the size conditioned Galton-Watson trees, having free parameters $w_1$ and $\beta$, where $w_n \sim n^{-\beta}$. A critical line separates the generic phase from the non--generic phase. For the values $2 < \beta \leq 3$ on the critical line (grey line) the offspring probability distribution has infinite variance and therefore this case corresponds to non--generic, critical trees.} \label{f:phase}
\end{figure}

The second application of Theorem \ref{thm:dh} concerns the attachment and grafting model which is a special case of the vertex splitting model. We give a novel proof that for the parameter range $D=3$, $\frac{1}{2} \leq \alpha \leq 1$ and for $D=\infty$, $\gamma = 0$, one has almost surely $d_h = 1/\alpha$ and $d_s = 2/(1+\alpha)$. This proves part of a previous conjecture \cite{stefansson:2011a}.

It remains an open problem to prove that the quenched Hausdorff and spectral dimension of the ag--model are almost surely $d_h = 1/\alpha$ and $d_s = 2/(1+\alpha)$ for the full range of the parameters $\alpha$, $\gamma$ and $D$. The only missing ingredient in the proof is to generalize Lemma \ref{l:agh} to the full range of parameters. We formulate this as the following:
\begin{conjecture} For the ag--ensemble $(\Gamma,\nu)$ it holds for any $\alpha$, $\gamma$ and $D$, that there exists a constant $c>0$ such that 
\begin  {equation}
 \measet{X^i_n>0}\leq \frac{c}{n}(1+o(1)).
\end  {equation}
\end{conjecture}
Furthermore, we expect that one even has the stronger result
$\measet{X^i_n>0} \sim 1/(\alpha n)$. This conjecture could possibly be proven by analysing the recursion relation (\ref{ag:hrec}) to find the critical behaviour of the generating function $H_R(\zeta)$ when $\zeta \rightarrow 1$ and $R \rightarrow \infty$. We hope to return to this proof in the near future. 

At this point we would also like to note that all of the results derived for the ag--model can straightforwardly be extended to Ford's $\alpha$-model \cite{ford:2005} and its generalisation, the $\alpha\gamma$--model \cite{chen:2008}. We leave a more detailed discussion for future work after the proof of the above conjecture.

There are also different applications of the presented work in the field of quantum gravity. Firstly, the results derived in the article for critical non-generic trees are useful to understand in more detail the branched polymer phase of Euclidean quantum gravity and, in particular, in the presence of matter. Furthermore, the critical non-generic trees are described by branching processes whose population size has polynomial growth which is potentially relevant to the analysis of recently introduced multigraph ensembles as models for four-dimensional causal or Lorentizan quantum gravity \cite{multiL,multiA}.

\subsection*{Acknowledgment}
The authors would like to thank Bergfinnur Durhuus, Georgios Giasemidis, Thordur Jonsson and John Wheater for discussions as well as the anonymous referee for useful comments which improved the manuscript. S.\"O.S.\ acknowledges hospitality at the University of Iceland. S.Z.\ acknowledges financial support of STFC grant ST/G000492/1. Furthermore, he would like to thank NORDITA for kind hospitality and financial support for a visit during which this work was initiated.

\addcontentsline{toc}{section}{References}

\def\cprime{$'$}
\providecommand{\href}[2]{#2}\begingroup\raggedright\endgroup

\end{document}